\documentclass[10pt,conference]{IEEEtran}

\newcommand{\ceiling}[1]{\left\lceil{#1}\right\rceil}
\newcommand{\setof}[1]{\left\{{#1}\right\}}
\newcommand{\set}[2]{\left\{#1\mid #2\right\}}
\newcommand{\spsp}{SP$^2$}
\newcommand{\ourprotocols}{Simultaneous Progressing Switching Protocols}

\usepackage{array}
\usepackage{graphicx}
\usepackage{amsmath}
\usepackage{amssymb}
\usepackage{cite}
\usepackage{setspace}
\usepackage{algorithm}
\usepackage[]{algorithmic}
\usepackage{tikz}
\usepackage{xcolor}
\usepackage{multirow}
\usepackage{paralist}
\usepackage{url}
\usepackage{etoolbox}
\usepackage{cancel}
\usepackage{tabu}
\usepackage{soul}
\usepackage{subfig}
\usepackage{amsthm}
\usepackage{bytefield}
\usepackage{mathtools}
\usepackage{marvosym}

\newtheorem{lemma}{Lemma}
\newtheorem{theorem}{Theorem}
\newtheorem{corollary}{Corollary}
\newtheorem{definition}{Definition}
\newtheorem{example}{Example}

\begin{document}
\thispagestyle{plain}

%
\title{Simultaneous Progressing Switching Protocols
  for Timing Predictable Real-Time Network-on-Chips}

%

\author{Niklas Ueter$^1$, Georg von der Br\"uggen$^1$, Jian-Jia Chen$^1$, Tulika
  Mitra$^2$, and Vanchinathan Venkataramani$^2$\\
  $^1$TU Dortmund University, Germany\\
  $^2$National University of Singapore, Singapore
}
\maketitle

\begin{abstract}
  Many-core systems require inter-core communication, and
  network-on-chips (NoCs) have been demonstrated to provide good
  scalability. However, not only the distributed structure but also
  the link switching on the NoCs have imposed a great challenge in the
  design and analysis for real-time systems. With scalability and
  flexibility in mind, the existing link switching protocols usually
  consider each single link to be scheduled independently, e.g., the
  worm-hole switching protocol. The flexibility of such link-based
  arbitrations allows each packet to be distributed over multiple
  routers but also increases the number of possible link states (the
  number of flits in a buffer) that have to be considered in the
  worst-case timing analysis.

  For achieving timing predictability, we propose less flexible
  switching protocols, called \emph{\ourprotocols{}} (\spsp{}), in
  which the links used by a flow \emph{either} all simultaneously
  transmit one flit (if it exists) of this flow \emph{or} none of them
  transmits any flit of this flow. Such an \emph{all-or-nothing}
  property of the SP$^2$ relates the scheduling behavior on the
  network to the uniprocessor self-suspension scheduling problem. We
  provide rigorous proofs which confirm the equivalence of these two
  problems. Moreover, our approaches are not limited to any specific
  underlying routing protocols, which are usually constructed for
  deadlock avoidance instead of timing predictability.  We demonstrate
  the analytical dominance of the fixed-priority \spsp{} over some of the
  existing sufficient schedulability analysis for fixed-priority
  wormhole switched network-on-chips.
\end{abstract}


\textbf{Keywords:} Network-on-Chip, Real-Time Scheduling, Simultaneous Progressing Switching
   Protocols

\section{Introduction}
\label{sec:introduction}

Power dissipation has constrained the performance scaling of
single-core systems in the past decade. Instead of increasing the
operational frequency of a processor, the chip manufacturers have
shifted their design focus towards chips with multiple or many cores
that operate at lower voltages and frequencies than their single-core
counterparts. In a multi- or many-core system, communication and
synchronization of the applications executed on different cores have
to be designed efficiently  from both hardware and software
perspectives.

The communication fabric of a multi- or many-core platform 
must scale with the number of cores. Otherwise, the computation capacity of the
cores may be wasted if they are waiting for the communication,
synchronization, or memory access. One possible approach to achieve
good scalability of the communication is the Network-on-Chip (NoC)
architecture, in which a switched network with routers is used to
provide the interconnection of the physical cores on a chip. The NoC
architecture allows parallel inter-core communication with moderate
hardware costs.  
NoCs are the prevalent choice of
interconnection due to their overall good performance and scalability
potential as reported by \mbox{Kavaldjiev et. al~\cite{Kavaldjiev}.}

The efficiency of a NoC highly depends on many design factors,
including topology, routing protocols, flow control, switching
arbitration protocols, etc. The currently available multi-core
platforms based on NoCs have employed different topologies, e.g., a
ring in the \emph{Intel Xeon Phi 3120A}, a \mbox{2-D} torus in \emph{MPPA Manycores} 
by Kalray, and a \mbox{2-D} mesh in \emph{Tilera TILE-Gx8036}. 
 Moreover, many different
communication protocols and communication topologies have been
proposed and evaluated in the literature.  Specifically,
\emph{flit-based} network-on-chips (NoCs) have been proposed with the
goal to decrease production cost and increase energy efficiency due to
less complex routers and decreased buffer sizes as compared to other
approaches.

Real-Time system design is concerned with the construction of systems
that can be formally verified to satisfy timeliness constraints. Such
real-time constraints are prevalent in, e.g., timing-sensitive
applications in embedded mobile platforms, automotive, and aerospace
applications.  To construct a hard real-time system on a NoC, each
hard real-time message (defined as an instance of a sporadic/periodic
flow) in the NoC has to be successfully transmitted from its source to
its destination before its deadline.

The approaches for real-time systems on a NoC apply two general strategies.  
One is to utilize time-division-multiplexing (TDM) 
to ensure that the timing constraints are satisfied by
constructing the transmission schedule statically with a repetitive
table, e.g., in \cite{Goosens2005, Paukovits2008ConceptsOS, Stefan2012, DBLP:journals/tvlsi/KasapakiSSMGS16, DBLP:conf/fpl/Schoeberl07, DBLP:conf/date/MillbergNTJ04, DBLP:journals/jsa/SchoeberlAAACGG15, DBLP:conf/rtns/HardeFBC18}.
 Another is to apply a
priority-based dynamic scheduling strategy in the routers (and
switches) to arbitrate the flits in the network, e.g., in~\cite{365629,627905,708526,1466499,Shi:RCA:2008,Kashif:2014,Kashif:2016,XIONG:2016,DBLP:journals/corr/IndrusiakBN16,Xiong:2017,Indrusiak:DATE:2018}. The
difficulty of the TDM strategy is to construct the TDM schedule and
the global clock synchronization, whilst the difficulty of the
priority-based scheduling strategy is to validate the schedulability, 
i.e., whether all messages can meet their deadlines in the worst case.

Specifically, for dynamic scheduling strategies, the wormhole-switched
fixed-priority NoC with preemptive virtual channels has recently been studied 
in a series of papers. 
The first attempts to tackle the schedulability
analysis were in 1994 in~\cite{365629} and 1997 in~\cite{627905}. Both
of them were found to be flawed in 1998 by Kim et al.~\cite{708526},
whose analysis was later found to be erroneous in 2005 by Lu et
al.~\cite{1466499}. The series of erroneous analyses
continued in~\cite{365629,627905,708526,1466499} until a seemingly
correct result by Shi and Burns~\cite{Shi:RCA:2008} published in
2008. Eight years later, Xiong et al.~\cite{XIONG:2016} pointed out
the analytical flaw in~\cite{Shi:RCA:2008} and disproved the safe
bounds in~\cite{Kashif:2014,Kashif:2016}. The erroneous patch in
\cite{XIONG:2016} was later fixed by the authors in their journal
revision in~\cite{Xiong:2017} in 2017. In the mean time, Indrusiak et
al.~\cite{DBLP:journals/corr/IndrusiakBN16,Indrusiak:DATE:2018}
presented new analyses, 
but they ``\emph{chose to provide intuitions, insight and
  experimental evidence on the proposed analysis and its improvements,
  rather than theorems or proofs}.''
They supported their analyses
by evaluating concrete cases, i.e., whether there was any observed
case which was claimed to meet the deadlines but in fact missed the
deadlines. However, such case studies cannot validate the correctness of their
analyses, as also stated by Indrusiak et
al. \cite{DBLP:journals/corr/IndrusiakBN16,Indrusiak:DATE:2018}.
Specifically, among the 10 results
\cite{365629,627905,708526,1466499,Shi:RCA:2008,Kashif:2014,Kashif:2016,XIONG:2016,DBLP:journals/corr/IndrusiakBN16,Xiong:2017}
published up to 2017, Table~VII~in~\cite{DBLP:journals/corr/IndrusiakBN16}
shows that eight of them were already disproved by counter examples, and two of them are \emph{probably}
safe as no counter examples have been given.  Nikov\'ic et
al.~\cite{Nikolic2019} published the most recent result for the
worst-case response timing analysis.

The fact that almost all proposed analyses for the problem 
discussed in
the previous paragraph have been found flawed suggests that the
scheduling algorithm and architecture can be potentially too complex
to be correctly analyzed. These approaches have adopted the well-known
worst-case response time analysis for uniprocessor sporadic real-time
tasks under fixed-priority uniprocessor scheduling developed
in~\cite{lehoczky89,joseph86responsetimes}, but they have never shown
the connection between worm-hole switching and uniprocessor
scheduling.

Another research line to analyze the worst-case response time of
wormhole-switched fixed-priority NoC with preemptive virtual channels
is to apply Network Calculus and Compositional Performance Analysis
(CPA), or their extensions, to analyze the transmission on the links
in a compositional manner, e.g.,
\cite{NC-Wormhole-WRR-2009,MPPA-ERTS-WCTT-2018,NoC-WCTT-RC-NC-WFCS-2016,Wormhole-BP-2015,DBLP:conf/date/TobuschatE17,DBLP:conf/date/RamboE15}. The
worst-case response time of a flow is the sum of the worst-case
response times on individual links, which can be
pessimistic. One exception is the analysis from
Giroudot~and~Mifdaoui~\cite{DBLP:conf/rtas/GiroudotM18}, which applies
the Pay Multiplexing Only Once principle by Schmitt et
al.~\cite{5755042}.

\textbf{Contributions:}
In this paper, we revisit the fundamental algorithmic
problem of flit-based NoC arbitration protocols with respect to
real-time constraints by hinting to the fundamental algorithmic
complexity. In addition, we identify the analytic pessimism of
wormhole-switched fixed-priority arbitration protocols due to the
large degree of uncertainty in system behaviour and thus hard to
verify timeliness properties.

This paper intends to answer a few
unsolved fundamental questions for real-time networking switching in a
NoC and has the following contributions:
\begin{itemize}
\item \emph{What is the fundamental difficulty of worst-case timing
    analysis when flit-based transmissions are handled by switch-based
    (link-based) scheduling?} We will show that the difficulty is
  mainly due to the explosion of the possible \emph{progression
    states} of a message in the transmission path. The existing
  analyses did not intend to prove the coverage of all possible
  progression states. Moreover, we will also argue why priority-based
  scheduling without controlling the number of possible progression
  states is therefore difficult to be analyzed and optimized.
\item \emph{Is there any equivalence between existing uniprocessor
    scheduling and NoC switching? } Yes, but, to the best of our
  knowledge, such a protocol has never been designed.  For achieving
  timing predictability, we propose less flexible switching protocols,
  called \emph{\ourprotocols{}} (\spsp{}), in which the links used by
  a message \emph{either} all simultaneously transmit one flit of this
  message \emph{or} none of them transmits any flit of this
  message. Such an \emph{all-or-nothing} property of the \spsp{} relates
  the scheduling behavior on the network to the uniprocessor
  self-suspension scheduling problem.  We provide rigorous proofs
  which confirm the equivalence of these two problems. 

\item We demonstrate the analytical dominance of the (work-conserving)
  fixed-priority version of our approach over the existing sufficient schedulability
  analyses for fixed-priority wormhole switched NoCs
  in~\cite{Xiong:2017,Indrusiak:DATE:2018}.  Moreover, our approaches
  are not limited to any specific underlying routing protocols, which
  are usually constructed for deadlock avoidance instead of timing
  predictability.


\end{itemize}


\section{System Model and Problem Definition}
\label{sec:model}

Network-on-Chips (NoCs) are characterized by the 
topology, routing protocol, arbitration, buffering, flow control mechanism, and
switching protocol.\footnote{Our notation of \emph{flows} is
  equivalent to \emph{tasks} and our notation of \emph{messages} is
  equivalent to \emph{jobs} in the classical notation of real-time
  systems community.} In this paper, we define a NoC as a collection
of cores $\textbf{A}$, routers ${\bf V}$, and links ${\bf \Lambda}$.
Each router is connected to at least one other router by two
physically separate links, i.e., up-link and
down-link. Figure~\ref{fig:noc-mesh-example} illustrates a meshed
network with $9$ cores, $\textbf{A}=\set{A_i}{i=1,2,\ldots,9}$, $9$
routers $\textbf{V}=\set{V_i}{i=1,2,\ldots,9}$, and $24$ links in
${\bf \Lambda}$ between $V_i$ and $V_j$ for some $i$ and $j$ and $18$
links between $A_i$ and $V_i$ for $i=1,2,\ldots,9$.
We assume that all the cores, routers, and links are homogeneous.
Therefore, the transmission rate and processing
capability are identical. 

\begin{figure}[t]
\centering
\includegraphics[width=0.35\textwidth]{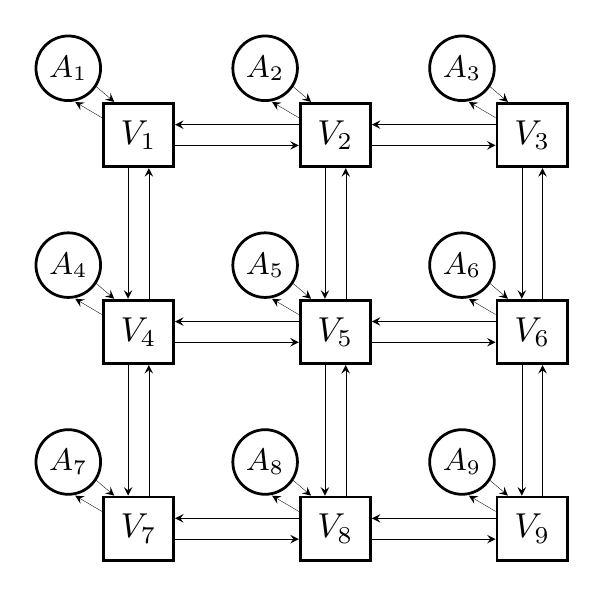}
\caption{Examplary 3x3 mesh NoC.
Each application is connected to a source-router where it injects messages.}
\label{fig:noc-mesh-example}
\end{figure}

\subsection{Switching Mechanisms}
\label{sec:switching}

With respect to switching, three different switching protocols have
been established, namely, circuit switching,
store-and-forward switching, and wormhole switching. 

\subsubsection{Circuit Switching} A packet (transmission unit) is
forwarded by the routers through dedicated routes that are 
reserved/allocated until the transmission is finished.  Therefore, each transmission can only
be preempted during the establishing of a route.  An advantage of this
approach is that
no buffers are required and subsequent arbitrary
deadlock-free routing, which allows for optimized and adaptive routing
schemes.  However, the overhead to establish the routes may
render this approach infeasible when small packets are injected frequently. 


\subsubsection{Store-And-Forward Switching} Routers can only forward
a packet once it is completely received and stored, which implies that
the routers must provide sufficient buffer capacity to store a
complete packet.  Fortunately, the arbitration protocol is suitable
for real-time analysis, since a packet may compete for
at most one link at each point in time. 

\subsubsection{Wormhole-Switching} Each packet is divided into smaller
transmission units, called flits, 
always including a designated
header and a designated tail flit which are used for control and routing.  That
is, each payload flit follows the output port of the header flit.  In
fixed-priority wormhole-switched 
NoCs, each router
contains virtual channels, i.e., separated buffers that contain flits
of a single packet.  Once the tail flit is transmitted and removed
from the buffer, the virtual channel can be used for flits of another
packet.  Furthermore,  the highest-priority flit is
scheduled to transmit over the link at each router.  In this approach, complete
packets do not have to be buffered and 
which allows smaller buffers
in the hardware design.  On the downside, each packet may be
distributed over multiple routers and subsequently compete for
multiple links at the same time which makes the timing analysis
complex.  Additionally, the limited number of virtual channels and
full buffers on the receiving router add additional interference which
complicate the analysis.

\subsection{Messages and Periodic/Sporadic Flows}
\label{sec:sporadic-periodic}

A periodic (sporadic) flow $f_i$ generates an infinite sequences of flow
instances, called messages, and has the following parameters: 

%
%
\begin{itemize}
  \item $T_i$ is the minimum inter-arrival time or period of the flow $f_i$,
  i.e., for a periodic flow one message is released exactly every $T_i$ time units and for a
  sporadic flow two subsequent messages are separated by at lest $T_i$.
\item ${\bf \Lambda}_i$ is the static routing path of the flow $f_i$, i.e.,
  $\lambda_{i1}, \lambda_{i2},\ldots, \lambda_{i\eta_i}$ is the sequence of
  the $\eta_i$ links that a message of $f_i$ has to be transmitted on. 
  We assume that a physical link cannot be used more than once
  in the static routing path ${\bf \Lambda}_i$ for any $i$.
\item $C_i$ is the worst-case number of flits of a message of
  $f_i$.
\item $D_i$ is the relative deadline of the flow $f_i$. That is, when
  a message is injected at time $t$, its absolute deadline is $t+D_i$.
  Our protocols are not restricted to any specific relation of the
  minimum inter-arrival time and the relative deadline $D_i$. However, our
  timing analysis will focus on the constrained-deadline cases, where $D_i \leq
  T_i$ $\forall i$.
\end{itemize}

\subsection{Problem Definition}
\label{sec:problem-definition}

The \emph{scheduler design} problem studied in this paper is defined
as follows: \emph{ We are given a NoC, defined as a collection of
  cores~$\textbf{A}$, routers~${\bf V}$, and links~${\bf \Lambda}$.
  For a given set $\textbf{F}$ of sporadic or periodic flows on the NoC,
  the objective is to design a switching mechanism (scheduling
  algorithm) that can ensure that all messages (instances of the
  flows) can meet their deadline.  }

The \emph{schedulability test} problem studied in this paper is
defined as follows: \emph{ We are given a NoC, defined as a collection
  of cores $\textbf{A}$, routers $\textbf{V}$, and links ${\bf \Lambda}$.
  For a given set $\textbf{F}$ of sporadic or periodic flows on the
  NoC and a switching mechanism, the objective is to validate whether
  the messages (instances of the flows) can meet their deadlines.}


We assume that the cores and routers are
synchronized perfectly with respect to time. That is, there is no
clock drift in the NoC. 
Otherwise, the clock drift must be considered carefully.
One solution is to introduce
additional delays and interferences to pessimistically bound the
impact due to clock drifts.  Moreover, we assume
\emph{discrete time}, i.e., the NoC operates in the granularity of a fixed
time unit and the finest granularity is the flit. 

\section{Existing Analytical Approaches for Worm-Hole Switching}
\label{sec:existing-and-progression-model}

In this section, we will first summarize the existing analytical
approaches for the worm-hole switching mechanism in
Section~\ref{sec:summary-of-existing-worm-hole}. Then, we will explain
the mismatch of the existing analyses and the underlying uniprocessor
scheduling in Section~\ref{sec:progession-model}.

\subsection{Summary of Existing Analyses}
\label{sec:summary-of-existing-worm-hole}

A first analytical approach to determine the worst-case response time
of sporadic traffic flows in wormhole-switched fixed-priority
network-on-chips was given by Mutka~\cite{365629} and
Hary~and~Ozguner~\cite{627905}. Both of them are based on the
schedulability analysis for uniprocessor sporadic real-time tasks
under fixed-priority scheduling developed
in~\cite{lehoczky89,joseph86responsetimes}. To analyze the worst-case
response time of the flow $f_i$, they considered the complete path
${\bf \Lambda}_i$ as a single shared resource, i.e., a uniprocessor. This
shared resource may not always  be  available for $f_i$, and they
modeled the unavailability by \emph{only} considering the higher-priority flows that use any link in
${\bf \Lambda}_i$, called \emph{direct
  interference}. They concluded that the problem is equivalent to the
fixed-priority uniprocessor scheduling, which was disproved by Kim et
al.~\cite{708526}, who 
showed that the flow $f_i$ can suffer
from the interference due to flow $f_j$ even if
${\bf \Lambda}_i$ and ${\bf \Lambda}_j$ have no intersection, called \emph{indirect interference}. 
By extending the notion of interference sets developed by Kim et al.~\cite{708526}, Lu et
al.~\cite{1466499} proposed to discriminate between flows that could
not interfere with each other to reduce the pessimism of the
analysis. 

However, both of the approaches in~\cite{708526, 1466499} assume that
the synchronous release of the first messages of the sporadic
real-time flows is the worst-case, i.e., similar to the critical instant
theorem in classical uniprocessor fixed-priority scheduling
proposed by Liu~and~Layland~\cite{liu73scheduling}. This statement was
later disproved in 2008 by Shi~and~Burns~\cite{Shi:RCA:2008}, where
jitter terms were added to model the asynchronous release of the
first messages of the sporadic real-time flows.
Based on the results of this work,
Kashif and Patel proposed a link-based analysis called stage-level
analysis~\cite{Kashif:2014, Kashif:2016} to achieve a tighter analysis. 
Both analyses were proved to be unsafe by Xiong et
al.~\cite{XIONG:2016} using simulations. It was discovered that a flit
of a higher-priority flow may induce interference more than once, i.e., on
multiple routers, thus rendering the conjectures made by Shi~and~Burns~\cite{Shi:RCA:2008} and
Kashif and Patel~\cite{Kashif:2014, Kashif:2016} false.  This behavior is referred to as multi-point progressive
blocking by Indrusiak et al.~\cite{Indrusiak:DATE:2018}.  The
state of the art with respect to fixed-priority wormhole-switched
networks-on-chips with infinite buffers is represented
by~\cite{DBLP:journals/corr/IndrusiakBN16, Xiong:2017}.
Unfortunately, the infinite buffer assumption is infeasible in real
systems, thus back-pressure effects that occur due to limited buffer
sizes in the routers have to be considered.  In the work of Indrusiak
et al.~\cite{Indrusiak:DATE:2018}, the authors incorporate buffer
sizes into the worst-case response time analysis. They ``\emph{chose to provide intuitions, insight and
  experimental evidence on the proposed analysis and its improvements,
  rather than theorems or proofs}.''
Thus, further counterexamples may be
found.  The fact that almost all proposed analyses have been found to
be flawed, suggests that the scheduling algorithm and architecture are
too complex to be reasonably analyzed. Further evidence for this claim
is that in the analyses provided by Indrusiak et
al.~\cite{Indrusiak:DATE:2018}, increased buffer sizes lead to
increased worst-case response times.
Nikol\'ic~et~al.~\cite{Nikolic2019} presented an improved analysis
over the results in
\cite{Indrusiak:DATE:2018,Xiong:2017}.

Motivated by this, we revisit the fundamental algorithmic problem of
packet-based network-on-chip scheduling and identify the analytic
pessimism incurred by the complexity of link-based arbitration as
harmful to routing and to verification. All the above results in
\cite{365629,627905,708526,1466499,Shi:RCA:2008,Kashif:2014,Kashif:2016,XIONG:2016,DBLP:journals/corr/IndrusiakBN16,Xiong:2017,Indrusiak:DATE:2018,Nikolic2019}
made an assumption that the schedulability analysis is somehow related
to a corresponding uniprocessor fixed-priority scheduling
problem. However, this has never been formally proved. We will explain
the mismatch by using the progression model in
Section~\ref{sec:progession-model}.

\subsection{Progression Model}
\label{sec:progession-model}

\begin{figure*}[t]
\centering
\includegraphics[width=0.6\textwidth]{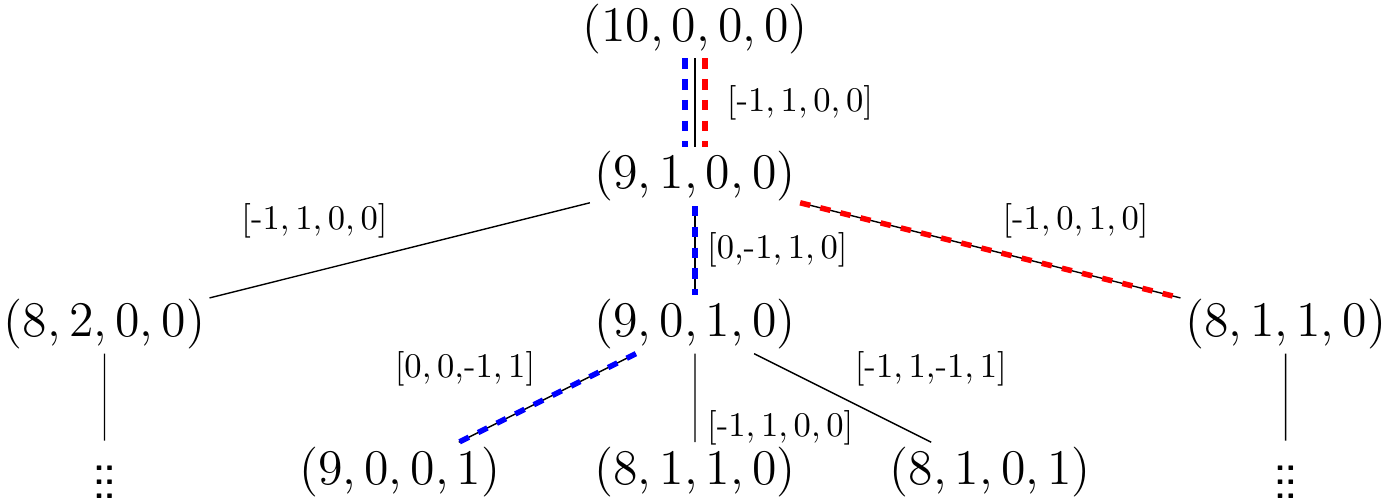}
\caption{Progressions of a message which involves 2 cores and 2
  routers, i.e., 3 links, when $C_i=10$. The numbers associated to an
  edge indicate the change of the buffers in the vector $\vec{B}_i$. 
  The red dashed path illustrates the beginning of a \emph{fastest} progression 
 and the blue dashed path illustrates the beginning of a \emph{slowest} progression.}
\label{fig:fsm-forwarding}
\end{figure*}

In this subsection, we detail why the link-based arbitration problem
does not match the uniprocessor scheduling model and illustrate the
subsequent problems in response-time analyses using uniprocessor
scheduling theory. To explain the mismatch, we will focus on  the
possible buffer states of one instance (i.e., message) of a flow $f_i$
under analysis.  Let $\vec{B}_i$ denote the state vector of the number
of flits that are buffered in the cores and routers involved in the
path ${\bf \Lambda}_i$.  Suppose that there are $\eta_i$ links involved in
${\bf \Lambda}_i$.  Note that the first element in $\vec{B}_i$ denotes the number
of flits of the whole message of $f_i$ in the source core to be sent and the
last element in $\vec{B}_i$ denotes the number of flits that have been
received at the destination core.

Recall our assumption in Section~\ref{sec:problem-definition} that the
NoC is assumed to be fully synchronized in time. 
Therefore, in each time unit, a buffered flit can be forwarded
to the next node (a core or a router).  Since the NoC works in
discrete time,  we can observe the changing of the vector $\vec{B}_i$
over time when considering only the time units at which the
message is sent.

When a flit is sent in a time unit at the $j$-th link in ${\bf \Lambda}_i$,
then the buffer of the $j$-th node is reduced by $1$ and the buffer of
the $j+1$-th node in ${\bf \Lambda}_i$ is increased by $1$.  For notational
brevity, let $\vec{y}_j$ be a vector of $\eta_i+1$ elements in which
all the elements are $0$ except the $j$-th element that is $-1$ and the
$j+1$-th element that is $1$. For example, $\vec{y}_1+\vec{y}_2$ implies
that the first and the second links send one flit forward in this time
unit. Moreover, $\vec{y}_1+\vec{y}_3$ implies that the first and the
third links send one flit forward in this time unit.

\begin{definition}[Progression]
\label{def:progression}
Consider a buffer state $\vec{B}_i$ at time $t$, in which all
elements in $\vec{B}_i$ are non-negative integers.  Suppose element
$z_j$ is either $0$ or $1$ for $j=1,2,\ldots, \eta_i$ and \textbf{at
  least one of them is $1$}.  Specifically, when $z_j$ is $1$, the
$j$-th link sends one flit forward.

Let $\vec{Y}$ be $\sum_{j=1}^{\eta_i} z_j\vec{y}_j$. For a buffer
state $\vec{B}_i$, $z_j$ for $j=1,2,\ldots,\eta_j$, and a vector $\vec{Y}$ defined above, the change of
the buffer state is \textbf{valid} if
\begin{itemize}
\item all elements in $\vec{B}_i+\vec{Y}$ are non-negative
  integers, and
\item the $j$-th element in $\vec{B}_i$ is $\geq z_{j+1}$ for $j=1,2,\ldots,\eta_i-1$.
\end{itemize}
If the change of the buffer state is valid, we say that the flow makes a
\textbf{progression} in this time unit, i.e., one clock cycle. \qed
\end{definition}

In a time unit, a link may or may not be utilized to send one flit of
$f_i$ in the switching mechanism. Therefore, there are $2^{\eta_i}$
combinations of the vectors of $\vec{y}$s.  Note that progressions do
not have to take place in two consecutive time units. If the message
is not sent in the next time unit at all, there is no progression of
the message.
As an illustrative example, consider $C_i=10$ and that the message is sent
from core $A_1$ via $R_1$ and $R_2$ to core $A_2$ in the NoC
illustrated in Figure~\ref{fig:noc-mesh-example}.  If one flit is sent
in a time unit, 
we get $\vec{B}_i = (C_i-1, 1,
0, 0)$. Now, there are three possibilities for the next time
unit when the NoC transmits a flit or multiple flits of the message:
\begin{itemize}
\item $\vec{B}_i = (C_i-1, 0, 1, 0)$: That is, $A_1$ does not send any
  flit but $R_1$ sends a flit to $R_2$. The progression is due to
  $\vec{Y}=(0,-1,1,0)$.
\item $\vec{B}_i = (C_i-2, 2, 0, 0)$: That is, $A_1$ sends one flit to
  $R_1$ but $R_1$ does not send a flit to $R_2$, i.e., 
  $\vec{Y}=(-1,1,0,0)$.
\item $\vec{B}_i = (C_i-2, 1, 1, 0)$: That is, $A_1$ sends one flit to
  $R_1$ and $R_1$ sends a flit to $R_2$, which means that the progression is due
  to
  \mbox{$\vec{Y}=(-1,1,0,0)+(0,-1,1,0)=(-1,0,1,0)$.}
\end{itemize}
The first three levels of the tree illustrated in
Figure~\ref{fig:fsm-forwarding} provides the above example. In each of
the above state, the next progression has to be considered. Due to
space limitation, we only further illustrate the progressions that are
possible when $\vec{B}_i$ is $(9, 0, 1, 0)$.

\begin{definition}[A Series of Progressions]
\label{def:series-progression}
A series of progressions is a sequence of progressions
defined in Definition~\ref{def:progression}, one after another,
starting from $\vec{B}_i = (C_i, 0, 0, \ldots, 0)$ to $\vec{B}_i = (0,
0, \ldots, C_i)$. \qed
\end{definition}

A safe analysis of the worst-case response time or the schedulability
for sending the message should consider all possible series of
progressions of $f_i$ starting from $\vec{B}_i = (C_i, 0, 0, \ldots,
0)$ to $\vec{B}_i = (0, 0, \ldots, C_i)$. If we only account for the
number of time units when the message of $f_i$ is transmitted, it is
not difficult to see that the \emph{slowest} one only sends one flit
forward per progression, in which the switching mechanism results in
$C_i\times \eta_i$ iterations of progressions. Moreover, the
\emph{fastest} one sends one flit (if available) forward for all
cores and routers involved in the path $\lambda_i$ per progression, in
which the switching mechanism results in $C_i + \eta_i-1$ iterations
of progressions.

All the results in
\cite{365629,627905,708526,1466499,Shi:RCA:2008,Kashif:2014,Kashif:2016,XIONG:2016,DBLP:journals/corr/IndrusiakBN16,Xiong:2017,Indrusiak:DATE:2018,Nikolic2019}
made an assumption that the corresponding uniprocessor scheduling
problem can use $C_i + \eta_i-1$ as the worst-case execution time of
the corresponding sporadic task to represent the flow $f_i$. This
assumption implicitly implies that the flow $f_i$ takes the fastest
series of progressions explained above. Such uniprocessor analyses are
only valid when the other iterations of progressions are accounted
correctly.
However, the fastest series of progressions for $f_i$ is not always
possible in the worst case. To ensure the correctness of the analysis,
some additional time units should be included.
Many patches have been provided to account for such additional
time units after the series of flaws found in 
\cite{365629,627905,708526,1466499,Shi:RCA:2008,Kashif:2014,Kashif:2016,XIONG:2016}.

Informally speaking, the researchers
in~\cite{365629,627905,708526,1466499,Shi:RCA:2008,Kashif:2014,Kashif:2016,XIONG:2016,DBLP:journals/corr/IndrusiakBN16,Xiong:2017,Indrusiak:DATE:2018,Nikolic2019}
have tried to construct their analyses by linking the problem to a
corresponding uniprocessor scheduling problem. Most of them were later
found flawed, e.g.,
\cite{365629,627905,708526,1466499,Shi:RCA:2008,Kashif:2014,Kashif:2016,XIONG:2016},
or without a formal proof, e.g.,
\cite{DBLP:journals/corr/IndrusiakBN16,Indrusiak:DATE:2018,Nikolic2019}.\footnote{The
  proofs in \cite{Nikolic2019} did not consider the equivalence of
  the worst-case response time analysis adopted in uniprocessor
  systems and the analysis of a NoC. Instead, they emphasized the
  quantification of different types of interferences. However, in many places
  in the proofs, e.g., the building blocks from Lemmas 3, 4, and 6 in
  \cite{Nikolic2019}, the derivation is based on examples. 
 } 
However, none of them has seriously considered all the possible progressions for transmitting $f_i$. 
Whether the worst-case response time analysis for
preemptive worm-hole switching is equivalent to
any specific form of uniprocessor scheduling problem remains as an
open question.

We strongly believe that the worst-case response time analysis or the
schedulability analysis for wormhole-switched fixed-priority
network-on-chip is highly complex, as the \emph{timing behavior is not
  uniprocessor equivalent}, as reported in the literature. In a
uniprocessor system, if a job is executed for $\delta$ time units, the
execution time of the job is reduced by $\delta$ time units. However,
sending $\delta$ flits can have different series of progressions in
the NoC.

If we would like to consider the complete path ${\bf \Lambda}_i$ as a single
shared resource, like in
\cite{365629,627905,708526,1466499,Shi:RCA:2008,Kashif:2014,Kashif:2016,XIONG:2016,DBLP:journals/corr/IndrusiakBN16,Xiong:2017,Indrusiak:DATE:2018,Nikolic2019},
and analyze the worst-case behavior by utilizing the corresponding
instance of the uniprocessor scheduling problem, the mapping from the
series of the progressions to the uniprocessor problem must be
\emph{formally proved}.

\emph{Please note that we do not claim that such a mapping is impossible.} We only
stated the mismatch.  Such mappings are potentially very difficult to
achieve precisely due to the large search space. However, safe
approximations and upper bounds are also missing in the literature.
In both cases, a correct proof should explain how to safely account
for the number of iterations in the progressions of the flows and map
them to the corresponding execution time in the constructed instance
of the uniprocessor scheduling problem.

\section{Simultaneous Progressing Switching Protocols}
\label{sec:simultaneous-progressing-switching-protocols}

To the best of our knowledge, there is no formal proof to demonstrate
the connection between the progressions of the messages on a NoC and
the corresponding instance of the uniprocessor scheduling problem. We
note that the analytical difficulty is potentially due to the
flexibility introduced in the switching mechanism. In the position
papers by Wilhelm et al.~\cite{DBLP:journals/tcad/WilhelmGRSPF09} and
Axer et. al~\cite{Axer:2014}, the authors state that system properties
that are subject to predictability constraints should already be
considered and guaranteed from the design. Since the worm-hole
switching protocol was not designed with  predictability
 constraints in mind,  
 designing new protocols that can be
safely analyzed without losing too much flexibility or efficiency can
be an alternative.


Instead of proving the complex scenarios in the standard worm-hole
switching, we 
propose another protocol which has only \emph{one
  series of progressions}. This less flexible switching protocols, called
\emph{\ourprotocols{}} (\spsp{}), achieves timing predictability by enforcing
that a flow $f_i$ is eligible to transmit on its
route \textbf{if and only if} it can be allocated for all links in ${\bf \Lambda}_i $
in-parallel.   
In other words, the links used by a flow
$f_i$ \emph{either} all simultaneously transmit one flit  of this flow (if it
exists) \emph{or} none of them transmits any flit of this flow. 
As a result, for a progression of $f_i$ in
a time unit, some links  in ${\bf \Lambda}_i $ may be reserved even though there is no 
flit to be transmitted over this link in this time unit,  
a behaviour similar to \emph{processor spinning}.

In order to formally define the schedules and analyze the schedulability, we use 
the following definition. 

\begin{definition}[Schedules]
\label{def:schedule-flow-links}
 A schedule $S_{\lambda_{j}}(t)$ is a function that 
 maps time $t$ in the time-domain to the flow that is 
 scheduled on the link $\lambda_{j}$ at that time. 
 Further, $S_{\lambda_{j}}(t) = \infty$ if $\lambda_j$ 
 is idle at time $t$. \qed
\end{definition}

We use the Iverson bracket to 
indicate whether a flow $f_i$ is scheduled on a link $\lambda_{j}$ 
at time $t$.\footnote{$[S_{\lambda_{j}}(t) = i]$ is $1$ if $f_i$ is scheduled on 
$\lambda_j$ at time $t$ and is $0$ otherwise.}
For convenience, we use $S_{{\bf \Lambda}_j}(t)$ to indicate 
the ordered set of schedules \mbox{$\lbrace S_{\lambda_{k}}(t)~|~\lambda_{k} \in {\bf \Lambda}_j \rbrace$}.
Moreover, we abbreviate our notation to a single value, i.e., 
$S_{{\bf \Lambda}_j}(t) = i$ to denote that all links in ${\bf \Lambda}_j$ 
schedule flow $f_i$ at time $t$.

\begin{definition}
  \label{def:SP2-a-bit-detail}
  A schedule $S$ implements the \spsp{} if for all $t \geq 0$, for each
  static routing path ${\bf \Lambda}_i$ of every flow $f_i$, the
  following implication holds:
\[
  (S_{\lambda_{k}}(t) = i \mbox{ for some }\lambda_{k} \in {\bf
  \Lambda}_i) \implies S_{{\bf \Lambda}_j}(t) = i
\]
\qed
\end{definition}

In general, the \spsp{} can be implemented with different
strategies.  The focus of this paper is not the implementation or
design of scheduling strategies to meet the schedule defined in
Definition~\ref{def:SP2-a-bit-detail}.  In the upcoming two
subsections, we discuss
two potential scheduling strategies that can be used to derive such schedules.

\subsection{Links to Gang Scheduling}  
To meet the deadline of a message of a flow $f_i$, that arrives at
time $t$, the concept of the \spsp{} requires to have $C_i+\eta_i-1$
time units to use all the links in $\Lambda_i$ simultaneously before
$t+D_i$. This is similar to the \textbf{rigid gang scheduling}
problem~\cite{DBLP:journals/lites/GoossensR16}, which can be defined
as follows:
\begin{itemize}
\item We are given a set of periodic/sporadic tasks to be executed on
  the given machines.
\item Each task has to simultaneously use a subset of the given
  machines as a \textbf{gang}. Whenever the task is executed, all of its required
  machines must be exclusively allocated to the task.
\end{itemize}
That is, we can consider that each of the links in ${\bf \Lambda}$
is a machine and each flow is a task. The links in $\Lambda_i$ form a
gang for flow $f_i$ and the execution time is $C_i+\eta_i-1$.

Finding the optimal schedule for the rigid gang scheduling problem has
been shown NP-hard in the strong sense even when all the tasks have
the same period and the same
deadline~\cite{Kubale:87:The-complexity-of-scheduling}. Moreover, even
special cases, like three
machines~\cite{BazewiczDell-Olmo:94:Corrigendum-to:-Scheduling} or
unit execution time per
task~\cite{HoogeveenVelde:94:Complexity-of-scheduling}, are also
NP-hard in the strong sense. The rigid gang scheduling problem for
implicit-deadline periodic real-time task systems (i.e., $D_i=T_i$ for
every task $\tau_i$) has been recently studied by
Goossens~and~Richard~\cite{DBLP:journals/lites/GoossensR16}. They
presented two algorithms, one is based on linear programming and
another is based on a heuristic algorithm. Moreover,
Harde~et~al.~\cite{DBLP:conf/rtns/HardeFBC18} constructed static
schedules by using a constrained-programming or an
integer-linear-programming (ILP) approach for harmonic real-time task
systems.
Another version of gang scheduling is the so-called \emph{global gang
  scheduling problem}
\cite{DBLP:journals/corr/RichardGK17,DBLP:conf/rtss/Dong017,DBLP:conf/rtss/KatoI09},
in which the set of machines used by a gang task is not fixed. A gang
task requires a certain amount of machines, and these machines can be
dynamically relocated at runtime. We note that the global gang
scheduling problem is \emph{unrelated} to the \spsp{}, since the links
used by a flow has to be defined from the source node to the
destination node.

\subsection{Work-Conserving Priority-Based Schedules}

Instead of optimizing the scheduling strategies in the routers in a
NoC for the \spsp{}, we will focus on the work-conserving priority-based \spsp{}. Such strategies can be dynamic-priority algorithms (i.e.,
two messages of flow $f_i$ may have different priorities) and
fixed-priority scheduling algorithms (i.e., all messages of flow $f_i$
have the same priority).

That is, each message has a priority. When two messages
intend to use one link at the same time, the higher-priority message
is scheduled and the lower-priority message is suspended. Whenever a
message is suspended in one of its links, it is suspended in all of
its links. 

We will focus on fixed-priority scheduling in  
the remainder of this paper. 
 We will focus on the theoretical benefit
of the \ourprotocols{} in Section~\ref{sec:scheduling-analysis}.

\section{\spsp{} Scheduling Analysis}
\label{sec:scheduling-analysis}

For a scheduling protocol $A$, a real-time schedulability analysis of a flow set
${\bf F}$ validates whether no flow in the flow set misses its deadline under any valid 
schedule generated by $A$. Such a validated 
flow set is hence called feasibly schedulable by $A$ or infeasible 
otherwise. Furthermore, an analysis of a schedulability test for an algorithm
$A$ is called \emph{sufficient} if all flow sets that are deemed to be feasibly schedulable 
by this test are actually feasibly schedulable. 
Likewise the test is called \emph{necessary}, 
if every flow set that is schedulable by algorithm $A$ is verified to be feasibly 
schedulable by the corresponding schedulability test. 
Hence, a \emph{necessary} and \emph{sufficient} schedulability test 
is called \emph{exact}. In the remainder of this paper, we will 
derive a sufficient schedulability test.

For each flow under analysis, we partition all other flows
into a direct contention domain and an indirect contention domain.
We define a (direct) contention domain of two flows by 
identifying the set of higher-priority flows that share at least 
one link and thus potentially directly interfere with each other.
Then, we consider the (indirect) contention domain of the flow under analysis, 
i.e., flows that do not directly share a link with this flow under analysis 
but interfere with flows in the (direct) contention domain.
In the remainder of this paper, we implicitly assume that the flow set 
is indexed such that a flow $f_i$ has higher priority than a flow 
$f_j$ if $i<j$.

In this section, we explain how the schedulability analysis for preemptive 
fixed-priority \ourprotocols{} can be related to the work-conserving fixed-priority preemptive 
uniprocessor self-suspension scheduling problem. 
In real-time scheduling theory, self-suspension denotes the property 
of an executable entity to be exempted 
from the scheduling decisions for a specified amount of time , i.e., 
\emph{self-suspension} time.  In our analysis, we use self-suspension 
to model a flow $f_j$ that is eligible to be scheduled using any work-conserving 
algorithm at a given time $t$ due to being the highest-priority 
flow and being active (released and not yet finished), 
but is not transmitted due to indirect contention. 
We prove that this (indirect) contention can be related to self-suspension behaviour 
in uniprocessor real-time scheduling theory and thus show 
the applicability of existing self-suspension aware schedulability analyses.
We formally define and prove the self-suspension equivalent property of 
fixed-priority scheduling using the \spsp{}.

We only briefly introduce the uniprocessor self-suspension problem 
and refer the reader to the existing
literature~\cite{Chen2018-suspension-review, 
ChenECRTS2016-suspension} 
for further information.
In short, self-suspension refers to the 
exemption of a ready schedulable entity from the scheduling decision 
for a certain amount of time. This exemption behavior is modeled as 
dynamic self-suspension and multi-segment self-suspension in the literature. 
In the former, the suspension pattern can be arbitrary and is only 
parametrically upper bounded by the total self-suspension time. 
This flexibility comes at the cost of more pessimism in the timing 
analyses. In the latter model, an upper bound of the duration and the 
number of a task's suspension intervals is known and can thus be used 
in the timing analyses. 
In this paper, we 
consider the dynamic self-suspension 
model and the corresponding timing analyses.

\begin{figure}[t]
\centering
\vspace{-0.2in}
\includegraphics[width=0.5\textwidth]{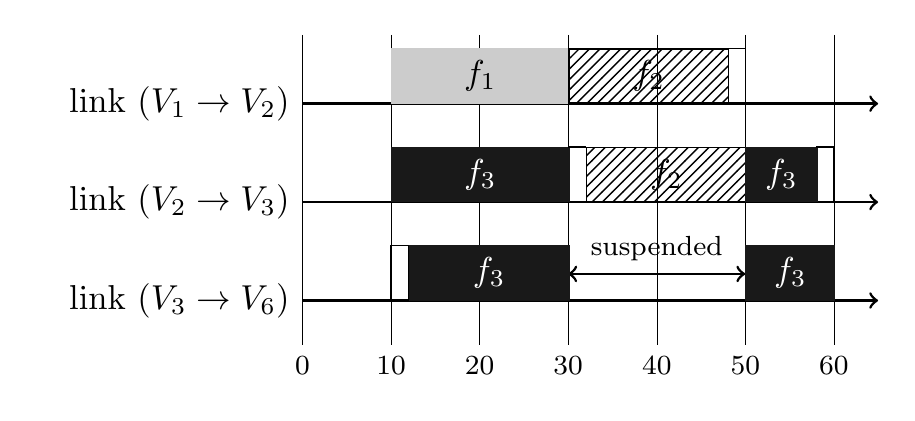}
\caption{An exemplary self-suspension instance for three flows 
using the \spsp{} and fixed-priority scheduling. The empty rectangles are
the time allocated for $f_2$ and $f_3$ but not used for
transmission since there is no available flit yet.}
\label{fig:self-suspension-example}
\end{figure}

\begin{example}[Self-Suspension Behaviour]
\label{example:self-suspension}
In Figure~\ref{fig:self-suspension-example}, 
three priority ordered flows $f_1, f_2 , f_3$ that transmit 
20, 19, and 29 flits respectively through the subset of routers 
$V_1, V_2, V_3$, and $V_6$ as illustrated Figure.~\ref{fig:noc-mesh-example}. 
Flow $f_1$ has the highest priority and $f_3$ has the lowest
priority. 
Flow $f_1$ transmits from $V_1$ to the router $V_2$, 
flow $f_2$ transmits from router $V_1$ through $V_2$ to $V_3$.
Moreover, flow $f_3$ transmits from $V_2$ through the 
router $V_3$ to $V_6$. In this example, all flows are released 
synchronously and are scheduled according to \ourprotocols{} 
using fixed-priority scheduling. 
Note that an actual transmission of a flit is denoted by a darkened 
box whereas the unfilled areas indicate the flows that are granted that link.
On the link from $V_1$ to $V_2$, flow $f_1$ 
precedes flow $f_2$ due to its higher priority.
Since $f_3$ does not share any link with $f_1$, flow
$f_3$ can transmit on its links.
After the finishing of $f_1$, flow $f_2$ attempts to 
transmit on its links and preempts flow $f_3$ 
on the link from $V_2$ to $V_3$. Due to the \spsp{}, 
$f_3$ is not eligible to transmit on the link from $V_3$ to $V_6$ 
despite being the highest-priority flow on that link.
Since the preemption by $f_2$ is transparent on that 
link, this behaviour is similar to the self-suspension 
property of real-time tasks. \qed
\end{example}

To formalize the direct contention domain of a flow $f_i$ under
analysis, we denote the set of higher-priority flows of $f_i$ that
share at least one link as $share_i = \lbrace f_j \in {\bf F}~|~ j <
i~, {\bf \Lambda}_j \cap {\bf \Lambda}_i \neq \emptyset \rbrace$.
Under the fixed-priority \spsp{}, each flow has a priority, assumed to
be unique here. In our analysis in this section, we assume
work-conserving fixed-priority arbitration:
\begin{itemize}
\item a flow transmits further if none of its link is used by any
  higher-priority flow, and
\item a flow does not transmit further if one of its link is allocated
  by another higher-priority flow.
\end{itemize}

Note that we analyze the schedulability of each flow individually 
under the assumption that the schedulability of all higher-priority 
flows has already been validated. In order to formally quantify 
the times $t$ such that a flow $f_j$ is eligible to transmit 
on all links in ${\bf \Lambda_i}$ in-parallel using 
\ourprotocols{} but is not transmitting, i.e., 
self-suspension behaviour, we give the following definition.

\begin{definition}
\label{def:lambda-i-self-suspended}
 A flow $f_j$ is said to be ${\bf \Lambda}_i$-self-suspended 
 at a time $t$, if the following conditions are satisfied:
 \begin{itemize}
  \item[Prop.~1:] $f_j$ is \emph{active}, i.e., released and not yet finished at time $t$
  \item[Prop.~2:] $f_j$ is the highest-priority flow on all the links in ${\bf \Lambda}_i$, i.e., $\min\{S_{{\bf \Lambda}_i}(t)\} > j$ 
  \item[Prop.~3:] ${\bf \Lambda}_i \supsetneq {\bf \Lambda}_j$ and ${\bf \Lambda}_i \cap {\bf \Lambda}_j \neq \emptyset$, i.e., ${\bf \Lambda}_j$ is 
    a true non-empty subset of ${\bf \Lambda}_i$
 \item[Prop.~4:] for all $\lambda_k \in {\bf \Lambda}_i$ $S_{\lambda_k}(t) \neq j$ at time $t$, i.e., $f_j$ is not scheduled on all the links \qed
 \end{itemize}
\end{definition}


Note that we require ${\bf \Lambda_i} \neq {\bf \Lambda_i}$, 
since the self-suspension like behaviour can only occur due to 
contention that is transparent to the flow under analysis. 
In the following theorem, we formally bound the set of 
flows that exhibit self-suspension behaviour in order to 
safely account for the additional interference.

\begin{lemma}
\label{lemma:incusion-lemma}
 Let a flow $f_j$ satisfy the properties $f_j \in share_i$ and 
 for all $f_n \in share_j : {\bf \Lambda}_i \cap {\bf \Lambda}_n \neq \emptyset$
 then it follows that $f_n \in share_i$. 
\end{lemma}

\begin{proof}
 This simply comes from the definitions. 
 By the first property, it follows that $j > i$ and ${\bf \Lambda}_j \cap {\bf \Lambda}_i \neq \emptyset$.
 The second property implies that $\forall f_n, n>j$ such that ${\bf \Lambda}_n \cap {\bf \Lambda}_j \neq \emptyset$, 
 the condition ${\bf \Lambda}_i \cap {\bf \Lambda}_n \neq \emptyset$ holds, which 
 implies that $f_n \in share_i$.
\end{proof}

\begin{theorem}
\label{thm:sup-lambda-self-suspending}
 In a schedule that is generated by any fixed-priority algorithm using \ourprotocols{}, 
 the set of flows that are ${\bf \Lambda}_i$-self-suspending is not larger than \newline 
 ${\bf SS(\Lambda_i)} = \lbrace f_{\ell} \in share_i ~|~ \exists f_n \in share_{\ell} :
 {\bf \Lambda}_n \cap {\bf \Lambda}_i = \emptyset \rbrace$.
\end{theorem}

\begin{proof}
 We prove this theorem by contrapositive, i.e., we show
 if $f_{j} \notin \lbrace f_{\ell} \in share_i~|~\exists f_n \in share_{\ell} : {\bf \Lambda}_n \cap {\bf \Lambda}_i = \emptyset \rbrace$ 
 then $f_j$ is not ${\bf \Lambda}_i$-self-suspending.
 Therefore, we must analyze the following two cases:
 \begin{enumerate}
  \item $f_j \notin share_i$,
  \item $f_j \in share _i~\text{and}~\forall f_n \in share_j : {\bf \Lambda}_i \cap {\bf \Lambda}_n \neq \emptyset$.
 \end{enumerate} 
 In the first case, let $f_j \notin share_i$ and thus by definition
 ${\bf \Lambda}_j \cap {\bf \Lambda}_i = \emptyset$. 
 Therefore, $f_j$ is not ${\bf \Lambda}_i$-self-suspending at any time $t$ by definition.
 
 In the second case, assume the existence of a time instant $t^*$ such that $f_j$ 
 is ${\bf \Lambda}_i$-self-suspended at time $t^*$ and satisfies the properties 
 stated in the second case. Then, $f_i$ is active at time $t^*$ 
 and for all $\lambda_k \in {\bf \Lambda_i}$ either $S_{\lambda_k}(t^*) = 0$ or $S_{\lambda_k}(t^*) > j$.
 Further, by the properties stated for the second case and the results from 
 Lemma.~\ref{lemma:incusion-lemma}, we know that flow $f_n \in share_i$.
 This is a contradiction, because no flow $f_n$ could have been active at time $t^*$ since 
 otherwise the schedule would have been $S_{\lambda_k}(t^*) = n$ for all $\lambda_k \in {\bf \Lambda}_i$ and $n < j$.
\end{proof}

For further analysis, it is required to bound the maximal 
amount of time a flow $f_j$ may be ${\bf \Lambda}_i$-self-suspending.

\begin{theorem}
 Let each higher-priority flow $f_j$ of the flow under analysis $f_i$ be 
 feasibly schedulable using \ourprotocols{}.
 Then, the cumulative amount of time that $f_j$ is ${\bf \Lambda}_i$-self-suspending 
 is at most $R_j-C_j$, where $R_j$ denotes the worst-case response-time of 
 flow $f_j$.
\end{theorem}

\begin{proof}
 We consider the following two cases:
\begin{enumerate} 
 \item Flow $f_j$ is not ${\bf \Lambda}_i$-self-suspending and thus 
   the self-trivially upper bounded by $R_j-C_j$.
 \item There exists at least one point in time $t \geq 0$
 such that $f_j$ is ${\bf \Lambda}_i$-self-suspending.
\end{enumerate} 
Let $S_j = \lbrace t \in [t_j, t_j+R_j)~|~f_j~\text{is}~{\bf \Lambda}_i\text{-self-suspending} \rbrace$ 
where $t_j$ denotes the release of a packet of $f_j$.
By Definition.~\ref{def:lambda-i-self-suspended}, we know that 
$f_j$ satiesfies Prop.~1 - Prop.~4 at time $t$ for all $t \in S_j$. 
Furthermore, since ${\bf \Lambda}_i \cap {\bf \Lambda_j} \neq \emptyset$ (Prop.~3) and 
by the \emph{all-or-nothing} property of \ourprotocols{}, 
it must be that \mbox{$[S_{{\bf \Lambda}_j}(t) = j] = 0$}.

Since by assumption, the schedulability of each higher-priority flow has already been validated, 
$f_j$ is feasibly schedulable, i.e., $\int_{t_j}^{t_j+R_j} 1 - [S_{{\bf \Lambda}_j}(\tau) = j]~d\tau = R_j - C_j$.
Due to the \spsp{} property, we know that all $t$ that satisfy
$(1-[(S_{{\bf \Lambda}_i}(t) = j]) = 1$ also satisfy $(1-[(S_{{\bf \Lambda}_j}(t) = j]) = 1$.
In conclusion, we have that 
$\int_{t \in S_j} 1-[(S_{{\bf \Lambda}_i}(\tau) = j]~d\tau \leq \int_{t \in S_j} 1-[S_{{\bf \Lambda}_j}(\tau) = j]~d\tau \leq R_j- C_j$ 
which concludes the proof.
\end{proof}

\begin{corollary}
\label{corollary:schedulable}
A sporadic constrained-deadline flow set ${\bf F}$ is fixed-priority schedulable using 
\ourprotocols{}, if for each flow $f_i$ the transformed higher-priority flow set:
\begin{equation}
f_j' =
\begin{cases}
     (C_j, T_j, D_j, S_j) & f_j \in {\bf SS(\Lambda_i)} \\
     (C_j, T_j, D_j) & \text{otherwise},
  \end{cases}
\end{equation}
is schedulable, where $S_j = R_j-C_j$.
\end{corollary}

The worst-case response time and schedulability of each flow 
$f_k$ has to be verified under the assumption that the higher-priority 
flows $f_1, f_2, \ldots, f_{k-1}$ are already verified to be schedulable.
Based on Corollary~\ref{corollary:schedulable}, any schedulability test 
that verifies the schedulability of sporadic constrained-deadline self-suspending task sets 
on uniprocessor systems with preemptive fixed-priority scheduling can be used, 
e.g., the state-of-the-art tests by Chen et al.~\cite{ChenECRTS2016-suspension}.

\section{Analytical Advantages of \spsp{}}
\label{sec:analytical-advantages}

In this section we shortly compare our fixed-priority \spsp{} and
schedulability analysis with some of the state-of-the-art
fixed-priority schedulability analyses for wormhole-switching NoC with
virtual channels proposed by Xiong et al.~\cite{Xiong:2017} and
Indrusiak et al.~\cite{Indrusiak:DATE:2018}. Unfortunately, we have to
admit to not be able to comprehend the analyses presented by Nikol\'ic
et al.~\cite{Nikolic2019} and are thus incapable to compare with their
methods analytically.  We note that we do not intend to directly compare the
state-of-the-art analyses due the different protocols and models but
to only examplify the analytical gains and benefits of the \spsp{}
compared to one-hop scheduling protocols.

Let $share_k$ be the set of higher-priority flows whose routing path
intersects with the routing path of $f_k$.  Let $share_k^1$ be $f_j
\in \{share_k \setminus \{ f_{\ell} \in share_k~|~share_{\ell}
\setminus share_k = \emptyset\}\}$. This notation follows~\cite{Nikolic2019}. In the plaintext, $share_k^1$ consists
of the flows in $share_k$ by excluding those flows $f_\ell$ in which
higher-priority flows that intersect with $f_\ell$, i.e., $share_\ell$
are also in $share_k$. This means that if $f_\ell$ is in $share_k^1$,
then $f_\ell$ is in $share_k$ and there exists one flow $f_n$ that is not
in $share_k$ but in $share_\ell$. Therefore, this notation is exactly
${\bf SS}({\bf \Lambda}_k)$ defined in
Theorem~\ref{thm:sup-lambda-self-suspending}, i.e., 
\[
{\bf SS}({\bf \Lambda}_k) = share_k^1
\]

According to Eq. (4) and Eq. (5) in~\cite{Nikolic2019}, the analyses
of the worst-case response time for \emph{preemptive worm-hole
  switching} from \cite{Xiong:2017,Indrusiak:DATE:2018} can be
computed by solving the minimum value $t>0$ of the following
function:\footnote{The term $J_j^R$ in Eq. (4) in \cite{Nikolic2019}
  is removed here since we consider sporadic flows without release jitter.}
\begin{equation}
  \label{eq:nikolic-baseline}
  \small
  t = C_k + \sum_{f_j \in share_k} \ceiling{\frac{t+ J_{j \rightarrow i}^I}{T_j}} \cdot (C_j + B_{j \rightarrow i})
\end{equation} where $B_{j \rightarrow i} \geq 0$ is the interference 
due to buffering, i.e., backpressure, 
and 
\begin{equation}\label{eq:nikolic-jitter}
J_{j \rightarrow i}^I = 
  \begin{cases}
    R_j - C_j &\mbox{ if } f_j \in share_k^1,\\
    0 &\mbox{ otherwise}.
  \end{cases}
\end{equation}

Now, we consider the following response time analysis from Chen et
al. in \cite{ChenECRTS2016-suspension}:\footnote{For notational
  consistency with \cite{ChenECRTS2016-suspension}, we here use the
  notation from \cite{ChenECRTS2016-suspension} for self-suspending
  task systems and assume that there are $k-1$ higher-priority flows
  in $share_k$. }
\begin{equation}\label{eq:suspension-test}
t = C_k + S_k+\sum_{i=1}^{k-1}\ceiling{\frac{t+Q_i^{\vec{x}}+(1-x_i)(R_i-C_i)}{T_i}} C_i 
\end{equation}
where $Q_i^{\vec{x}} = \sum_{j=i}^{k-1} (S_j \times x_j)$ and for a
certain binary assignment of $x_i \in \setof{0,1}$ for
$i=1,2,\ldots,k-1$. 



According to Corollary~\ref{corollary:schedulable} and ${\bf SS}({\bf
  \Lambda}_k) = share_k^1$, we have $S_j=R_j-C_j$ if $f_j \in share_k^1$ and $S_j=0$ if $f_j \notin share_k^1$. Moreover, since $f_k$ is not in ${\bf SS}({\bf
  \Lambda}_k)$, we have $S_k=0$.  If $f_j$ is in $share_k^1$, we set $x_j$ to $0$; otherwise, we set $x_j$ to
$1$. In such a setting of $\vec{x}$, the value $Q_i^{\vec{x}}$ in
Eq.~\eqref{eq:suspension-test} is always $0$. Together with $S_k=0$ by
definition and the definition in Eq.~\eqref{eq:nikolic-jitter}, the
analysis in Eq.~\eqref{eq:suspension-test} becomes:

{\footnotesize \begin{align} 
t  = &C_k + \sum_{f_j \in {\bf SS}({\bf
  \Lambda}_k)}
    \ceiling{\frac{t+R_j-C_j}{T_j}} C_j + \sum_{f_j \in
      share_k\setminus {\bf SS}({\bf
  \Lambda}_k)} \ceiling{\frac{t}{T_j}}
    C_j \nonumber\\
= &C_k + \sum_{f_j \in share_k^1}
    \ceiling{\frac{t+R_j-C_j}{T_j}} C_j + \sum_{f_j \in
      share_k\setminus share_k^1} \ceiling{\frac{t}{T_j}}
    C_j \nonumber\\
    = & C_k + \sum_{f_j \in share_k} \ceiling{\frac{t+J_{j
          \rightarrow i}^I}{T_j}} C_j\nonumber\\
    \leq & C_k + \sum_{f_j \in share_k} \ceiling{\frac{t+J_{j
          \rightarrow i}^I}{T_j}} (C_j+B_{j\rightarrow i}) = \mbox{RHS
      of Eq.\eqref{eq:nikolic-baseline}}\nonumber
\end{align}}

Therefore, the worst-case response time analysis from
\cite{Xiong:2017,Indrusiak:DATE:2018} is dominated by our analysis
from Corollary~\ref{corollary:schedulable} by applying
suspension-aware response time analysis from Chen et
al.~\cite{ChenECRTS2016-suspension}.

\section{Implementation Considerations}
\label{sec:implementation-considerations}

The architectural implementation of the priority-based \ourprotocols(\spsp{}) providing the 
\emph{all-or-nothing} property requires to rethink previous router designs.
In the state-of-the-art wormhole switching protocols, the decision 
at each router is local, i.e., each router simply chooses the highest-priority 
flow on any outgoing link. 
In contrast, the \emph{all-or-nothing} property requires 
global decision making. 

The \spsp{} is a general concept, and there could be different
possible realizations.  One possibility is to use centralized
arbitration which decides and dispatches the priority information to
the routers. However, this may incur high hardware cost.
Possible implementations are subject of future research efforts 
and beyond the conceptual scope of this paper.

\section{Conclusion}
\label{sec:conclusion}

In this paper, we discuss the fundamental difficulty of worst-case
timing analysis of flit-based pipelined transmissions over multiple
links in-parallel in network-on-chips.  The space of possible
\emph{progression states} that need to be covered by an analysis hints
to the mismatch with uniprocessor scheduling theory and their
assumptions, thus making analyses complex and prone to being
optimistic.  To that end, we propose \emph{\ourprotocols{}} (\spsp{}),
in which the links used by a message \emph{either} all simultaneously
transmit one flit of this message \emph{or} none of them transmits any
flit of this message. For this family of protocols, we formally prove
the matching with uniprocessor scheduling assumptions and theory.
Furthermore, we show the relation between the uniprocessor
self-suspension scheduling problem and the \spsp{} scheduling and
provide formal proofs to confirm this relation. In addition, we provide a
sufficient schedulability analysis for fixed-priority \spsp{}
scheduling. 

We note that the existing link-based switching and the proposed
\spsp{} are in fact two extreme scenarios with respect to the series of
progressions. The \spsp{} approach always results in the fastest
series of progressions in any case by sacrificing the possibility to
send part of the messages even when only one link is blocked by
another higher-priority message/flow.  The link-based switching
mechanism (worm-hole protocol) allows the flexibility to send only one
flit of a flow forward at a time unit in the NoC, but it has to
potentially consider the slowest series of progressions of the flow in
the worst case. It may be possible to design timing predictable
systems with good average-case performance by pruning unnecessary
progressions that have to be considered in the protocol. However, this
alternative was not considered in this paper. 

We strongly believe that a timing-predictable switching protocol in
NoCs should be carefully designed so that NoC-based many-core systems
can yield predictable performance. This paper provides protocols that
can be implemented with different strategies.  We believe that our
proposals can be a first step towards predictable switching protocols
of NoCs. In our future work, we will explore possible design options
and their tradeoffs in the schedulability analyses and design
complexity/cost. 


\clearpage

\bibliographystyle{abbrv}
\bibliography{real-time.bib}

\end{document}